\documentclass[]{amsart}
\usepackage{amsfonts}
\usepackage{amsmath}
\usepackage{amssymb}
\usepackage{graphicx}
\usepackage{float}
\usepackage{hyperref}
\usepackage{geometry}
\newtheorem{theorem}{Theorem}[section]
\newtheorem{proposition}[theorem]{Proposition}
\newtheorem{remark}[theorem]{Remark}
\newcommand{\sint}{\textstyle{\int}}  
\newcommand{\C}{\mathbb{C}}

\newcommand{\R}{\mathbb{R}}
\newcommand{\Z}{\mathbb{Z}}

\newcommand{\I}{\mathrm{i}}
\renewcommand{\P}{\mathbb{P}}

\newcommand{\lp}{\left(}
\newcommand{\rp}{\right)}

\newcommand{\diag}{{\rm diag}}
\newcommand{\Cg}{\mathbb{C}^g}
\newcommand{\siegel}{\mathbb{H}^g}
\newcommand{\Zg}{\mathbb{Z}^g}
\newcommand{\Rg}{\mathbb{R}^g}
\newcommand{\Hg}{\mathbb{H}^g}

\newcommand{\ernst}{\mathcal{E}}

\newcommand{\thetapq}{\Theta_{\mathbf{pq}}}
\newcommand{\vp}{\mathbf{p}}
\newcommand{\vq}{\mathbf{q}}
\newcommand{\vz}{\mathbf{z}}
\newcommand{\vn}{\mathbf{n}}
\newcommand{\vm}{\mathbf{m}}
\newcommand{\B}{\mathbb{B}}
\renewcommand{\L}{\mathcal{L}}

\begin{document}
\title{On a class of algebro-geometric solutions to the Ernst equation}
\author{Eddy B. de Leon}
    \email{Eddy-Brandon.De-Leon-Aguilar@u-bourgogne.fr}
\address{Institut de Math\'ematiques de Bourgogne,
		Universit\'e de Bourgogne, 9 avenue Alain Savary, 21078 Dijon
		Cedex, France}
\date{\today} 
\thanks{This work was partially supported by the EIPHI Graduate School (contract ANR-17-EURE-0002)the Bourgogne
Franche-Comt\'e Region, the European fund FEDER, and by the European Union Horizon 2020 research and innovation program under the Marie Sklodowska-Curie RISE 2017 grant agreement no. 778010 IPaDEGAN. The author thanks D. Korotkin for helpful discussions and hints.}
\begin{abstract}
     We discuss a class of solutions to the Ernst equation in terms 
	of theta functions with characteristics. We 
	show that it is necessary to take into account a phase factor, which arises from a shift by a lattice vector, and impose conditions on the characteristics of the theta functions in order for the presented function to solve the Ernst equation for all the considered parameters. 
\end{abstract}

\maketitle

\section{Introduction} 
The theory of multi-dimensional theta functions has proven to be 
useful in many domains of mathematics and physics. In the theory of 
integrable systems, they provide quasi-periodic solutions to nonlinear 
evolution equations such as the Korteweg-de Vries (KdV) and nonlinear 
Schr\"odinger equations, see \cite{BB} for a historic account and 
many references. Similar applications in general relativity have 
different properties. Although the Einstein field equations are in 
general non-integrable, they are if the spacetime under study has at 
least two commuting Killing vectors (symmetries). In particular, the 
vacuum stationary axisymmetric Einstein equations in Weyl coordinates can be 
reduced to the integrable Ernst equation  \cite{E}, which is 
\begin{equation} \label{ernst_eq_gen}
    \Re(\ernst) \Delta\ernst = (\nabla \ernst)^2,
\end{equation}
where the Laplacian $\Delta$ and the gradient $\nabla$ are the usual 
operators in cylindrical coordinates $(\rho,\zeta,\varphi)$ and since we consider axisymmetric spacetimes, the Ernst potential $\ernst$ 
does not depend on the azimuthal coordinate $\varphi$.
This means that solutions of the Einstein equations  with the aforementioned symmetries can be constructed via quadratures of the solutions of the Ernst equation.
This approach allows the description of various rotating spacetimes, the most prominent of which is the Kerr spacetime which is interpreted as a rotating black hole (see \cite{KR} for further discussions), as well as counter-rotating dust disks \cite{NM, KR}.
 
A large class of solutions to the Ernst equation was originally given by Korotkin \cite{K} in 
terms of multi-dimensional Riemann theta functions.
Alternatively, 
these solutions can be expressed in terms of theta functions with 
characteristics, see \cite{KKS}. These solutions are constructed over a family of genus-$g$ hyperelliptic 
curves (the precise definitions are 
given in Section \ref{SP}). Namely, the dependence of these solutions 
on the physical coordinates is via the modular dependence of the 
periods of the Riemann surface. Thus, the physical coordinates do not 
enter directly the argument of the theta functions as is the case for 
evolution equations, such as the
KdV and Kadomtsev-Petviashvili equations which are 
given on a Riemann surface independent of the physical coordinates. 
In the latter case, the solutions show quasi-periodicity properties, 
whereas this is not the case for
solutions to the Ernst equation. 
Both the dust disk solution and the Kerr solution are constructed on a family of hyperelliptic curves of genus $2$, but the Kerr solution is obtained as a limiting case, the so-called solitonic limit \cite{KR}.

It is the purpose of this note to show that in terms of the complex variable $\xi=\zeta+\I\rho$ (where $\zeta$ and $\rho$ are the physical coordinates), the potential
\begin{equation} \label{Ernst_pot0}
    \ernst(\xi,\bar{\xi}) = e^{-\pi\I \sum_j p_j } \frac{\thetapq\left( \sint_\xi^{\infty^+},\B_\xi \right)}{\thetapq\left( \sint_\xi^{\infty^-},\B_\xi \right)},
\end{equation}
solves the Ernst equation (thus an Ernst potential), where $\sint_\xi^{\infty^\pm}$ and $\B_\xi$ are quantities parametrized by $\xi$ and $\thetapq$ is the multi-dimensional theta function \eqref{theta_pq} with fixed arbitrary characteristics $\vp\in\Rg$ and $\vq\in\Cg$, whose components satisfy the reality conditions 
\begin{equation} \label{reality}
    \Re(q_j) = \left\{ \begin{array}{ll}
      \frac{1}{2} \sum_{k\neq j} p_k, & \text{if } E_j=\bar{F}_j, \\
    -\frac{1}{4}+\frac{1}{2} \sum_k p_k, & \text{if } E_j, F_j\in\R,
    \end{array} \right.
\end{equation}
where $E_j$, $F_j$ are the branch points of the defining hyperelliptic curves \eqref{hyper}. These terms are properly defined in Section \ref{SP}. There are two differences with respect to the potential presented in \cite{KKS}. First, it is necessary to include a phase factor, which arises upon shifting the argument of the theta functions by a lattice vector when the characteristic $\vp$ is different from zero. Second, the reality conditions on the characteristics had to be modified to \eqref{reality}. However, the class of solutions \eqref{Ernst_pot0} coincides with the one presented in \cite{KKS} when $\sum_j p_j$ is an integer and all the branch cuts are of the form $E_j=\bar{F}_j$.

This note is organized as follows. In Section \ref{SP}, we recall 
some basic facts regarding the Ernst equation, hyperelliptic curves 
and theta functions, as well as Fay's identity. In Section \ref{SEE}, 
we show that the potential \eqref{Ernst_pot0} solves the Ernst 
equation by using the approach presented in \cite{KKS}.

\section{Preliminaries} \label{SP}
\subsection{Ernst equation}
The line element of a stationary axisymmetric vacuum spacetime in the 
Weyl-Lewis-Papapetrou form reads
\begin{equation}
    ds^2 = -f(dt+Ad\varphi)^2+f^{-1}\left[e^{2k}(d\zeta^2+d\rho^2)+\rho^2 d\varphi^2 \right],
\end{equation}
where $k=k(\rho,\zeta)$, $A=A(\rho,\zeta)$, $f=f(\rho,\zeta)$. The Einstein equations in these coordinates are equivalent to the Ernst equation together with the relation $f=\Re(\ernst)$ and the quadratures of
\begin{align}
  A_\xi = 2\rho \frac{(\ernst-\overline{\ernst})_\xi}{(\ernst+\overline{\ernst})^2}, \quad k_\xi = 2\I\rho \frac{\ernst_\xi \ernst_{\bar{\xi}}}{(\ernst+\overline{\ernst})^2}.
\end{align}

It is convenient to consider the Ernst equation in the complex coordinates, in which it takes the form
\begin{equation} \label{ernst_eq}
    \lp \ernst+\overline{\ernst} \rp \lp \ernst_{\xi\bar{\xi}} - \frac{1}{2(\bar{\xi}-\xi)} \left( \ernst_{\bar{\xi}} - \ernst_{\xi} \right) \rp = 2 \ernst_{\xi} \ernst_{\bar{\xi}} ,
\end{equation}

\subsection{Hyperelliptic curves}

Recall that the period matrix $\B$ of any smooth algebraic curve $\L$ is 
defined as the matrix with components $\B_{jk}=\sint_{b_j}\omega_j$, 
where $\{a_1,b_1,...,a_g,b_g\}$ is a canonical basis of the first 
homology group and $\{ \omega_1,...,\omega_g\}$ is a basis of 
holomorphic differentials normalized with respect to the $a_j$ 
cycles, i.e., $\sint_{a_j}\omega_k=\delta_{jk}$. The period matrix $\B$ is known to be
complex symmetric and with positive definite imaginary part. This matrix defines the full-rank lattice  $\Lambda_{\B} = \Z^g+\B\Z^g$ in $\C^g$.

An algebraic curve $\L$ defines the Abel map $\mathcal{A}: p \mapsto \sint_{p_0}^p \omega$, 
for a base point $p_0\in\L$. Although $\mathcal{A}(p)$ depends on the path, it is unique in $\C^g$ modulo $\Lambda_{\B}$. For simplicity, we use the notation $\sint_{p_1}^{p_2}:= \mathcal{A}(p_2) - \mathcal{A}(p_1) $. Notice that the difference is independent of the base point $p_0$. 

Consider the family of hyperelliptic curves
\begin{equation} \label{hyper}
    \L_\xi = \{ (x,y)\in\C^2 | y^2=(x-\xi)(x-\bar{\xi}) \prod_{j=1}^g (x-E_j)(x-F_j) \},
\end{equation}
parametrized by $\xi\in\C$, with distinct branch points $E_j,F_j\in\C$ and the pairwise condition that either $E_j=\bar{F}_j$ or $E_j,F_j\in\R$.

The basis cycles $\{a_j,b_j\}$ of the first homology group are chosen such that the action of the holomorphic involution $\tau(x,y)=(\bar{x},\bar{y})$ on them is 
\begin{equation} \label{eq:cycles_condition}
\begin{split}
    \tau(a_j) &= -a_j, \quad \text{for } j=1,\ldots, g,\\
    \tau(b_j) &= \left\{ \begin{array}{ll}
     b_j + \sum_{k\neq j} a_k, & \text{if } E_j=\bar{F}_j,\\
    b_j + \sum_{k} a_k, & \text{if } E_j,F_j\in\R. \end{array} \right.
\end{split}
\end{equation}
This guarantees that the real part of the period matrix $\B_\xi$ of $\L_\xi$ will have half-integer coefficients independent of $\xi$, as shown in Appendix \ref{appendix:cycles}.
Figure \ref{fig:cut_system_pq} shows cycles satisfying such conditions, where $a_j$ are the cycles encircling the branch cuts $[E_j,F_j]$ in counterclockwise direction and $b_j$ are those going from $[\xi,\bar{\xi}]$ to $[E_j,F_j]$ on the $+$-sheet. In the following, the $\pm$ scripts indicate whether we are considering the $+$ or $-$ covering sheet of $\L_\xi$. 

\begin{figure} [H]
    \centering
    \includegraphics[width=10cm]{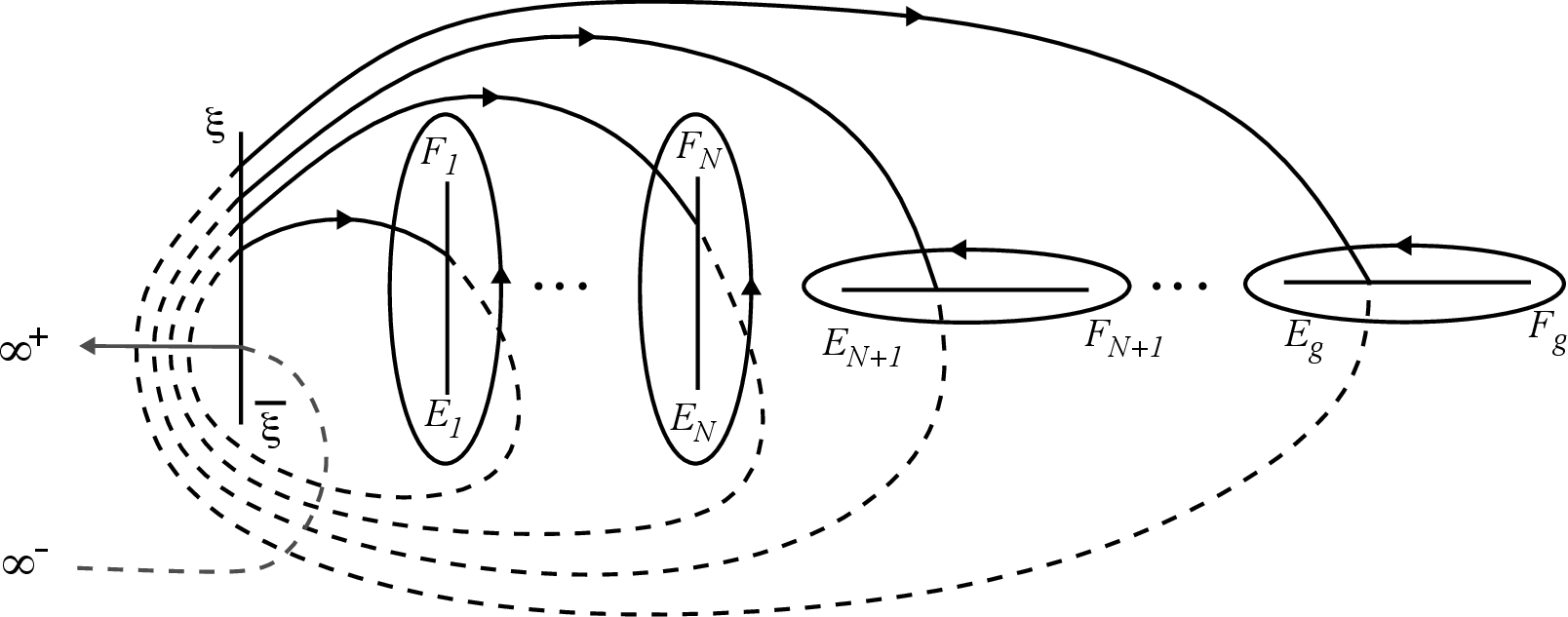}
    \caption{Choice of cycles.}
    \label{fig:cut_system_pq}
\end{figure}
The path $\gamma_\infty$ from the point $\infty^-$ to $\infty^+$ is chosen such that
\begin{equation} \label{eq:path_inf}
    \tau(\gamma_\infty) = \gamma_\infty - \sum_k a_k.
\end{equation}

\subsection{Theta functions}
The multi-dimensional theta function with characteristics $\vp\in\R^g$ and $\vq\in\C^g$ is defined by the series
\begin{equation} \label{theta_pq}
    \thetapq(\vz,\B):=\sum_{\vn\in\Zg} \exp\lp \pi \I \langle\vn+\vp,\B(\vn+\vp)\rangle + 2\pi \I \langle\vn+\vp,\vz + \vq \rangle \rp,
\end{equation}
for any $\vz\in\C^g$ and $\B\in\Hg$, where $\Hg$ is the space of $g\times g$ complex symmetric matrices with positive definite imaginary part. The latter implies that $\thetapq$ is an entire function in $\C^g$. This function is quasi-periodic with respect to the lattice $\Lambda_{\B}$ and its translation by the lattice vector $\vm\in\Zg$ is given by
\begin{equation} \label{theta_period}
    \thetapq(\vz+\vm,\B) = e^{2\pi\I \langle \vp, \vm\rangle } \thetapq(\vz,\B).
\end{equation}
Moreover, theta functions can be defined on a curve $\L$ via its period matrix $\B$ and the Abel map. Namely,
$$\Theta: \L \to \C, \quad p\mapsto \thetapq(\mathcal{A}(p)+v,\B),$$
for any $v\in\Cg$.

Summarizing, to every point $\xi\in\C$ we associate the period matrix $\B_\xi$ of the curve $\L_\xi$ as well as the Abel maps $\sint_{\xi}^{\infty^\pm}$, which enter as the arguments of the theta functions.
\subsection{Fay identity}
It is a relation between points on a compact Riemann surface $\mathcal{R}$, e.g., the hyperelliptic curves \eqref{hyper}. This functional relation is given in terms of the theta functions with characteristics defined by the period matrix of $\mathcal{R}$ and it holds on arbitrary points $a,b,c,d \in \mathcal{R}$, for all $\vz\in\Cg$.
\begin{equation} \label{fay_id}
\begin{split} 
    E(c,a)E(d,b) & \thetapq(\vz+\sint_b^c)\thetapq(\vz+\sint_a^d)+E(c,b)E(a,d) \thetapq(\vz+\sint_a^c) \thetapq(\vz+\sint_b^d) \\
    & = E(c,d)E(a,b)\thetapq(\vz)\thetapq(\vz+\sint_a^d+\sint_b^c),
\end{split}
\end{equation}
where $E(x,y)=\Theta_{\vp^*\vq^*}(\sint_y^x)/[h_\Delta(x)h_\Delta(y)]$ is the prime form, with the spinor $h_\Delta(x)$ satisfying 
$$h_\Delta^2(x)=\sum_{j=1}^g \frac{\partial \Theta_{\vp^*\vq^*}(0)}{\partial z_j} \omega_j(x),$$
see \cite{mumford, fay} for further discussions.
The theta function $\Theta_{\vp^*\vq^*}$ is required to have a non-singular odd half-integer characteristic, i.e., $2\vp^*,2\vq^*\in \Z^{g}/(2\Z^{g})$ such that $\Theta_{\vp^*\vq^*}(0)=0$ and $\nabla \Theta_{\vp^*\vq^*}(0)\neq 0$.

%
\section{Solution to the Ernst equation} \label{SEE}
In this section, we show that the Ernst potential \eqref{Ernst_pot0}
solves the Ernst equation \eqref{ernst_eq} with fixed arbitrary characteristics $\vp\in\R^g$, $\vq\in\C^g$ satisfying the reality conditions \eqref{reality}. The proof is divided in three steps. First, we show that the complex conjugate of \eqref{Ernst_pot0} can be expressed in terms of theta functions corresponding to the same period matrix $\B_\xi$. Second, we show that the real part of the Ernst potential can be simplified via the Fay identity. Third, we show that the proof presented in \cite{KKS}, which holds for $\vp=0$, can be extended to any $\vp\in\Rg$.

We are interested in expressing the complex conjugate of the Ernst potential in terms of theta functions of the same matrix $\B_\xi$, whose arguments must be represented as the Abel maps of points on $\L_\xi$, in order to use the Fay identity \eqref{fay_id}. This is given by the following proposition.
\begin{proposition} \label{prop:conjugate}
Let $\ernst(\xi,\bar{\xi})$ be the potential defined by \eqref{Ernst_pot0} with characteristics $\vp\in\R^g$, $\vq\in\C^g$ satisfying the reality conditions \eqref{reality}.
Then, its complex conjugate is 
\begin{equation}  \label{eq:conj_ernst_1}
    \overline{\ernst(\xi,\bar{\xi})} = e^{-\pi\I \sum_j p_j } \frac{\thetapq\left( \sint_{\bar{\xi}}^{\infty^+},\B_\xi \right)}{\thetapq\left( \sint_{\bar{\xi}}^{\infty^-},\B_\xi \right)}.
\end{equation}
\end{proposition}

\begin{proof}
From the definition of the multi-dimensional theta function \eqref{theta_pq}, it can be observed that
$$\overline{\thetapq(\vz,\B)} = \alpha \cdot \thetapq(-\bar{\vz}-2\Re(\vq+\B \vp)+\diag(\Re(\B)),\B),$$
for all $\vz\in\Cg$ and for any matrix $\B\in\Hg$ satisfying the condition $2\Re(\B)\in M_{g\times g}(\Z)$, where $\alpha\in\C$ is a constant independent of $\vz$ (see Appendix \ref{App_conj}). This condition is satisfied by the period matrices $\B_\xi$ of hyperelliptic curves of the form \eqref{hyper} with the choice of cycles shown in Figure \ref{fig:cut_system_pq}, as shown in Appendix \ref{appendix:cycles}. 

Moreover, the conditions \eqref{reality} on the $\vp$, $\vq$ characteristics are equivalent to the vanishing of the term $-2\Re(\vq+\B \vp)+\diag(\Re(\B))$. Then,
$$\overline{\thetapq(\vz,\B_\xi)} = \alpha \cdot \thetapq(-\bar{\vz},\B_\xi),$$
which implies
\begin{equation} \label{eq:conj_ernst_2}
    \overline{\ernst(\xi,\bar{\xi})} = e^{\pi\I \sum_j p_j } \frac{\thetapq\left( - \overline{\sint_{\xi}^{\infty^+}},\B_\xi \right)}{\thetapq\left( - \overline{\sint_{\xi}^{\infty^-}},\B_\xi \right)}.
\end{equation}

The next step is expressing the complex conjugates $\overline{\sint_\xi^{\infty^\pm}}$ in terms of the Abel maps of $\bar{\xi}$. Notice that $\bar{\vz}=2\Re(\vz)-\vz$ for any $\vz\in\C^g$ and in particular,
\begin{align*}
   \overline{\sint_\xi^{\infty^\pm}} = 2\Re \left( \sint_\xi^{\infty^\pm} \right) - \sint_\xi^{\infty^\pm}.
\end{align*}
On the other hand, the integrals $\sint_\xi^{\infty^\pm}$ can be written in terms of $\sint_{\infty^-}^{\infty^+}$, whose real part is given explicitly by $\Re(\sint_{\infty^-}^{\infty^+})=(\frac{1}{2},\ldots,\frac{1}{2})$ in Appendix \ref{appendix:cycles}.  Notice that $\sint_\xi^{\infty^+}=-\sint_\xi^{\infty^-} $, since the paths of the integrals $\sint_\xi^{\infty^+}$ and $\sint_\xi^{\infty^-}$ have the same projection on $\C\P^1$. Thus,

$$\sint_{\infty^-}^{\infty^+}=\int_{\infty^-}^{\xi}+\int_{\xi}^{\infty^+} = - \int^{\infty^-}_{\xi}+\int_{\xi}^{\infty^+} = \pm 2 \int_{\xi}^{\infty^\pm}.$$

Implying,
\begin{equation} \label{eq:conj_abel}
\overline{\sint_\xi^{\infty^\pm}} = - \sint_\xi^{\infty^\pm} \pm \Re( \int_{\infty^-}^{\infty^+} ) .
\end{equation}

The integral $\int_{\bar{\xi}}^\xi$ is one half of the integral along the cycle encircling the cut $[\bar{\xi},\xi]$ in clockwise direction, which is equivalent to $\sum_k a_k$. Then, 
\begin{equation*}
    \sint_{\bar{\xi}}^\xi = \frac{1}{2} \sum_k \int_{a_k} = \Re( \int_{\infty^-}^{\infty^+} ) = (\frac{1}{2},\ldots,\frac{1}{2}).
\end{equation*}
Thus, substituting this value and \eqref{eq:conj_abel} into \eqref{eq:conj_ernst_2} we obtain
\begin{equation} 
    \overline{\ernst(\xi,\bar{\xi})} = e^{\pi\I \sum_j p_j } \frac{\thetapq\left( \sint_{\bar{\xi}}^{\infty^+}-2\sint_{\bar{\xi}}^\xi,\B_\xi \right)}{\thetapq\left( \sint_{\bar{\xi}}^{\infty^-},\B_\xi \right)}.
\end{equation}

Finally, the form \eqref{eq:conj_ernst_1} of the complex conjugate of the Ernst potential follows by considering the translation
\begin{equation*}
    \thetapq\left( \sint_{\bar{\xi}}^{\infty^+} - 2 \sint_{\bar{\xi}}^\xi,\B_\xi \right) = e^{ -2\pi\I \langle \vp, 2 \sint_{\bar{\xi}}^\xi \rangle } \thetapq\left( \sint_{\bar{\xi}}^{\infty^+} ,\B_\xi \right)
\end{equation*}
by the lattice vector $2 \sint_{\bar{\xi}}^\xi = (1,\ldots,1)$, which is obtained with the property \eqref{theta_period} of theta functions.
\end{proof}

This proposition implies that the real part of the Ernst potential \eqref{Ernst_pot0} can be expressed in a simple form, since all the involved theta functions are now in terms of the same period matrix $\B_\xi$, which allows us to use Fay's identity \eqref{fay_id}. For ease of readability, we omit the second argument $\B_\xi$ in the sequel. Thus, using \eqref{eq:conj_ernst_1} from Proposition \ref{prop:conjugate}, we obtain
\begin{align*}
    \ernst(\xi,\bar{\xi})+ \overline{\ernst(\xi,\bar{\xi})} = e^{-\pi \I \sum_j p_j } \left[ \frac{ \thetapq(\sint^{\infty^+}_{\xi}) \thetapq(\sint^{\infty^-}_{\bar{\xi}}) + \thetapq(\sint^{\infty^+}_{\bar{\xi}}) \thetapq(\sint^{\infty^-}_{\xi}) }{ \thetapq(\sint^{\infty^-}_{\xi}) \thetapq(\sint^{\infty^-}_{\bar{\xi}})}  \right],
\end{align*}
and considering Fay's identity \eqref{fay_id} with $\vz=0$, $a=\xi$, $b=\bar{\xi}$, $c=\infty^-$, $d=\infty^+$; the lemma
$$ \frac{E(\infty^-,\xi)E(\infty^+,\bar{\xi})}{E(\infty^-,\bar{\xi})E(\xi,\infty^+)} = 1,$$
which is proven in \cite{KKS}; and the property $E(x,y)=-E(y,x)$ of the prime form, we observe that
\begin{align*}
    \thetapq(\sint^{\infty^+}_{\xi}) \thetapq(\sint^{\infty^-}_{\bar{\xi}}) + &\thetapq(\sint^{\infty^+}_{\bar{\xi}}) \thetapq(\sint^{\infty^-}_{\xi}) = \\
    & \frac{E(\infty^-,\infty^+)E(\xi,\bar{\xi})}{E(\xi,\infty^-)E(\bar{\xi},\infty^+)} \thetapq(0)\thetapq(\sint_\xi^{\infty^+}+\sint_{\bar{\xi}}^{\infty^-}).
\end{align*}
Therefore, the real part of the Ernst potential is
\begin{equation} \label{ernst_real}
    \Re(\ernst(\xi,\bar{\xi})) = Q \cdot e^{-\pi \I \sum_j p_j } \frac{\thetapq(0) \thetapq(\sint_{\bar{\xi}}^\xi)}{\thetapq(\sint_{\xi}^{\infty^-}) \thetapq(\sint_{\bar{\xi}}^{\infty^-})},
\end{equation}
where 
$$Q = Q(\xi,\bar{\xi}) = \frac{1}{2} \frac{E(\infty^-,\infty^+)E(\xi,\bar{\xi})}{E(\xi,\infty^-)E(\bar{\xi},\infty^+)} = \frac{\Theta(\sint_\xi^{\infty^-}) \Theta(\sint_{\bar{\xi}}^{\infty^-})}{\Theta(0) \Theta(\sint_\xi^{\bar{\xi}})}.$$
The latter equality is obtained from the fact that $\ernst \equiv 1$ if $\vp,\vq=0$.

Finally, with formula \eqref{ernst_real}, we show that the potential \eqref{Ernst_pot0} solves the Ernst equation.

\begin{theorem}
    Let $\ernst(\xi,\bar{\xi})$ be the potential defined by \eqref{Ernst_pot0} with fixed arbitrary characteristics $\vp\in\Rg$, $\vq\in\Cg$ satisfying the reality conditions \eqref{reality}. Then, it solves the Ernst equation \eqref{ernst_eq} in complex coordinates, which can be written as
\begin{equation} \label{ernst_eq_2}
    (\ernst+\overline{\ernst}) \Delta \ernst = 8 \ernst_\xi \ernst_{\bar{\xi}}.
\end{equation}
\end{theorem}
\begin{proof}
Since the phase factor in \eqref{Ernst_pot0} is independent of $\xi$ if the basis cycles of the first homology group are chosen such that they satisfy conditions \eqref{eq:cycles_condition}, the derivatives $\ernst_\xi$, $\ernst_{\bar{\xi}}$ and the Laplacian $\Delta \ernst$ are just those shown in \cite{KKS} multiplied by this phase factor. Namely,
\begin{align*}
    \ernst_\xi & = \frac{1}{2} c_2(\infty^-,\xi,\infty^+) e^{-\pi \I \sum_j p_j }  \frac{\thetapq(0)}{\thetapq^2(\sint_\xi^{\infty^-})} D_\xi \thetapq(0),\\
    \ernst_{\bar{\xi} } & = \frac{1}{2} c_2(\infty^-,\bar{\xi},\infty^+) e^{-\pi \I \sum_j p_j } \frac{\thetapq(\sint_{\bar{\xi}}^\xi)}{\thetapq^2(\sint_\xi^{\infty^-})} D_{\bar{\xi}} \thetapq(\textstyle{\sint}^{\bar{\xi}}_\xi),
\end{align*}
\begin{align*}
    \Delta\ernst = \frac{1}{Q} c_2(\infty^-,\xi,\infty^+) c_2(\infty^-,\bar{\xi},\infty^+) e^{- \pi \I \sum_j p_j } \frac{\thetapq(\sint_{\bar{\xi}}^{\infty^-})}{\thetapq^3(\sint_{\xi}^{\infty^-})} D_{\bar{\xi}} \thetapq(\textstyle{\sint}^{\bar{\xi}}_\xi ) D_\xi \thetapq(0),
\end{align*}
where the coefficients $c_2(\infty^-,\xi,\infty^+)$, $c_2(\infty^-,\bar{\xi},\infty^+)$ are functions defined in terms of prime forms \cite{KKS}. Therefore, with these values for $\ernst_\xi$, $\ernst_{\bar{\xi}}$ and $\Delta \ernst$, and considering the form \eqref{ernst_real} for $\Re(\ernst)$, it follows that the potential \eqref{Ernst_pot0} solves the Ernst equation.
\end{proof}

\begin{remark}
If the exponential factor in \eqref{Ernst_pot0} is not taken into account, this class of solutions is only valid if $\sum_j p_j$ is an integer, e.g., $\vp=0$. Otherwise, formula \eqref{ernst_real} does not hold.
\end{remark}

\begin{remark}
An equivalent form of the class of solutions \eqref{Ernst_pot0} is 
\begin{equation} 
    \ernst(\xi,\bar{\xi}) = e^{-\pi\I \sum_j p_j } \frac{\thetapq\left( \sint_\xi^{\infty^+}+\frac{1}{2}\Delta,\B_\xi \right)}{\thetapq\left( \sint_\xi^{\infty^-}+\frac{1}{2}\Delta,\B_\xi \right)},
\end{equation}
with fixed arbitrary $\vp\in\Rg$, $\vq\in\Cg$ subject to the reality condition $\Re(\vq)+R\vp = 0$, where $R$ is given by \eqref{eq:real_B} and $\Delta:=\diag(R)$. 
\end{remark}

%
%
\appendix
\section{Complex conjugate of theta functions} \label{App_conj}
\begin{proposition}
    Let $\B$ be any Riemann matrix $\B\in\siegel$ whose real part $R:=\Re(\B)$ has half-integer coefficients. Then, the complex conjugate of its associated theta function with characteristics \eqref{theta_pq} can be written as
\begin{align} \label{eq:conj_theta}
    \overline{\thetapq(\vz,\B)} = \alpha \cdot \thetapq(-\bar{\vz}-2\Re(\vq+\B \vp)+\diag(R),\B),
\end{align}
where $\alpha\in\C$ is a constant that does not depend on $\vz$. 
\end{proposition}

\begin{proof}
In general, from the definition \eqref{theta_pq} of multi-dimensional theta functions with characteristics $\vp\in\R^g$ and $\vq\in\C^g$, it can be observed that their complex conjugate can be written as
\begin{align*} 
    \overline{\thetapq(\vz,\B)} = \thetapq(-\bar{\vz}-\bar{\vq}-\vq,-\bar{\B}) = \thetapq(-\bar{\vz}-2\Re(\vq),\B-2R),
\end{align*}
for all $\vz\in\C^g$ and $\B\in\Hg$. Moreover, since $R$ is assumed to have half-integer coefficients,
$$\begin{pmatrix}
   \mathbb{I}_g & -2R \\
    0 & \mathbb{I}_g
\end{pmatrix}\in \mathrm{Sp}(2g,\Z), $$ 
implying that $\B-2R$ is symplectically equivalent to $\B$. Thus, using the modular transformation of theta functions we obtain \eqref{eq:conj_theta}.
\end{proof}

\section{Choice of cycles} \label{appendix:cycles}
We choose cycles such that the action of the holomorphic involution $\tau(x,y)=(\bar{x},\bar{y})$ on them satisfies \eqref{eq:cycles_condition}. Figure \ref{fig:cut_system_pq} shows an example of such cycles. These conditions imply that the real part of the period matrices of $\L_\xi$ are half-integers and that they do not depend on the parameter $\xi\in\C$. Indeed, due to the condition $\tau(a_j) = -a_j$, the action $\tau^*$ of $\tau$ on the normalized differentials is $\tau^*(\omega_k)=-\bar{\omega}_k$ (see \cite{BB}). Therefore, the complex conjugate of the period matrices are of the form
\begin{align*}
    \bar{\B}_{jk} &= \int_{b_j} \bar{\omega}_k = - \int_{b_j} \tau^*(\omega_k) = -\int_{\tau(b_j)} \omega_k ,\\
    &= -\int_{b_j} \omega_k - \sum_{l} \int_{a_l} \omega_k + \sigma_j \int_{a_j} \omega_k,\\
    &= -\B_{jk} - (1-\sigma_j \cdot \delta_{jk}),
\end{align*}
where $\sigma_j=1$ if $E_j=\bar{F}_j$ and $\sigma_j=0$ if $E_j,F_j\in\R$. Thus, the components of the real part of the period matrices $\B_\xi$, which we denote $R=\Re(\B_\xi)$, are independent of $\xi$ and their explicit values are
\begin{equation} \label{eq:real_B}
    R_{ij} = \left\{ \begin{array}{ll}
        0, & \text{if } i=j \text{ and } E_j=\bar{F}_j, \\
        -\frac{1}{2}, & \text{otherwise}.
    \end{array} \right.
\end{equation}

Notice that the reality conditions \eqref{reality} can be equivalently expressed as
\begin{equation*}
    \Re(\vq)+R\cdot \vp = \frac{1}{2} \diag(R).
\end{equation*}

On the other hand, since the path $\gamma_p$ for the integral $\sint_{\infty^-}^{\infty^+}\omega$ is chosen such that the action of the holomorphic involution $\tau$ is given by \eqref{eq:path_inf}, the complex conjugate of $\sint_{\infty^-}^{\infty^+}\omega_j$ is
\begin{align*}
    \overline{ \sint_{\infty^-}^{\infty^+} \omega_j} =  \sint_{\gamma_\infty} \bar{\omega}_j &= - \sint_{\gamma_\infty} \tau^*(\omega_j) = - \sint_{\tau(\gamma_\infty)} \omega_j,\\
    &=- \sint_{\gamma_\infty} \omega_j + \sum_k \sint_{a_k} \omega_j = - \sint_{\infty^-}^{\infty^+} \omega_j +  1,
\end{align*}
which implies
\begin{equation} \label{eq:real_infty_pm}
    \Re \left( \sint_{\infty^-}^{\infty^+} \omega \right) = \frac{1}{2} \sum_k \sint_{a_k} \omega = \left( \frac{1}{2}, \ldots, \frac{1}{2} \right).
\end{equation}

%
%
\bibliographystyle{amsplain}

\end{document}